\documentclass[runningheads]{llncs}

\usepackage {amsfonts}
\usepackage{amsmath,amssymb}
\usepackage{graphicx}
\usepackage{verbatim}
\usepackage{subfigure}
\usepackage[ruled,longend]{algorithm2e}
\usepackage{enumerate}
\spnewtheorem{fct}{Fact}{\bf}{\it}

\let\doendproof\endproof
\renewcommand\endproof{~\hfill\qed\doendproof}

\titlerunning{   }
\authorrunning{ }
{\bf}{\it}

\begin{document}


\titlerunning{Bar 1-Visibility Drawings of 1-Planar Graphs}
\authorrunning{S. Sultana {\it et al.}}
\title{Bar 1-Visibility Drawings of 1-Planar Graphs}

\author{
Shaheena\ Sultana \and Md.\ Saidur Rahman  \and 
Arpita Roy\and 
Suraiya Tairin 
} 

\institute{
Graph Drawing and Information Visualization Laboratory,\\
Department of Computer Science and Engineering,\\
Bangladesh University of Engineering and Technology\\
\email{ shaheenaasbd@yahoo.com, saidurrahman@cse.buet.ac.bd, arpita116@yahoo.com, suraiya\_pakhi@yahoo.com}
}

\pagenumbering{arabic}                                                                                                                                                                                                         
\maketitle

\begin{abstract}
A bar $1$-visibility drawing of a graph $G$ is a drawing of $G$ where each vertex is drawn as a horizontal line segment called a bar, each edge is drawn as a vertical line segment where the vertical line segment representing an edge must connect the horizontal line segments representing the end vertices and a vertical line segment corresponding to an edge intersects at most one bar which is not an end point of the edge. A graph $G$ is  bar 1-visible if $G$ has a bar 1-visibility drawing. A graph $G$ is 1-planar if $G$ has a drawing in a 2-dimensional plane such that an edge crosses at most one other edge.  In this paper we give linear-time algorithms to find bar 1-visibility drawings of diagonal grid graphs and maximal outer 1-planar graphs. We also show that recursive quadrangle 1-planar graphs and pseudo double wheel 1-planar graphs are bar $1$-visible graphs.


\end{abstract}                                                       

\section{Introduction}
\label{Introduction} 
A {\it $1$-planar drawing} of a graph $G$ is a drawing of $G$ on a two dimensional plane where an edge can be crossed by at most another edge. A graph $G$ is {\it $1$-planar} if $G$ has a 1-planar drawing. A {\it straight-line drawing} of a graph $G$ is a drawing of $G$ such that every edge of $G$ is drawn as a straight-line segment. A {\it right angle crossing drawing} or {\it RAC drawing} is a straight-line drawing where any two crossing edges form right angles at their intersection
point. A {\it RAC graph} is a graph that has a RAC drawing. A {\it bar $1$-visibility drawing} of a graph $G$ is a drawing of $G$ where each vertex is drawn as a horizontal line segment called a bar, each edge is drawn as a vertical line segment where the vertical line segment representing an edge must connect the horizontal line segments representing the end vertices and a vertical line segment corresponding to an edge intersects at most one bar which is not an end point of the edge. A graph $G$ is a bar {\it $1$-visible} if $G$ has a bar 1-visibility drawing. A bar $1$-visible graph and a bar $1$-visibility drawing of the same graph is shown in Figures~\ref{figure:bvg1} (a), and (b), respectively. 

\begin{figure}[!h]
\centering
\includegraphics[width=0.7\textwidth]{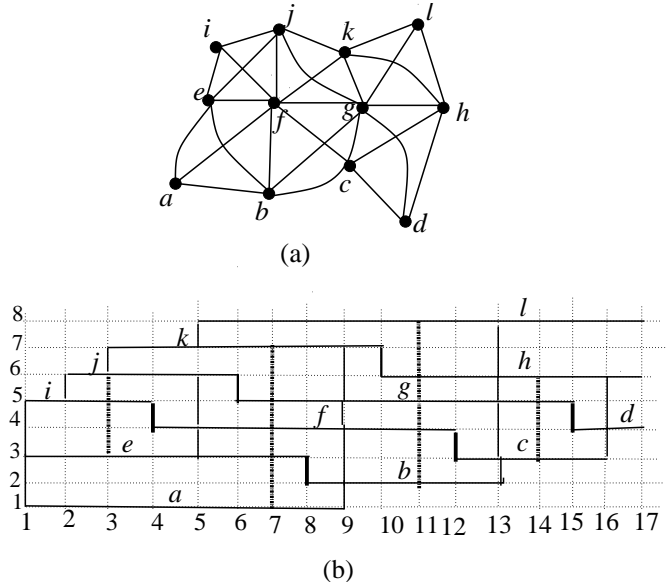}
\caption{(a) A  bar $1$-visible graph, and (b) a bar 1-visibility drawing of the same graph.} 
\label{figure:bvg1}
\end{figure}

Recently 1-planar graphs and RAC graphs have been extensively studied \cite{AFKMT10,DEL10,FM07,S10}. 
Pach and Toth proved that 1-planar graphs with $n$ vertices have at most $4n-8$ edges \cite{PT97}, whereas Didimo {\it et al.} 
showed that a RAC graph with $n>3$ vertices has at most $4n-10$ edges \cite{DEL09}. 
Recognizing the both classes of graphs are 
NP-hard \cite{ABS11,KM08}. Recently Eades and Liotta studied the relationship between dense RAC graphs and dense 1-planar graphs \cite{EL12}.

The concept of bar visibility
drawing came up in the early 1980s when many new problems in visibility theory
arose, originally inspired by applications dealing with determining visibilities
between different electrical components in VLSI design. Other applications
arise when large graphs are to be displayed in a transparent way, and in the
rapidly developing field of computer graphics. A \textit {bar visibility drawing }of a planar graph $G$ is a drawing of $G$ where each vertex is drawn as a horizontal line segment and each edge is drawn as a vertical line segment where the vertical line segment representing an edge must connect the horizontal line segments representing the end vertices.
Otten and Wijk~\cite{OV78} have shown that every planar graph admits a visibility drawing, and  Tamassia and Tollis~\cite{TT86} have given a linear-time algorithm for constructing a visibility drawing of a planar graph.  Dean {\it et al.} have introduced a generalization of visibility drawing for a non-planar graph which is called bar $k$-visibility drawing~\cite{DEGLST07}. In a bar $k$-visibility drawing of a graph a horizontal line corresponding to a vertex is called a {\it bar}, and the vertical line segment corresponding to an edge intersects at most $k$ bars which are not end points of the edge. Thus a visibility drawing is a bar $k$-visibility drawing for $k = 0$.

In this paper we give  linear-time algorithms to find  bar 1-visibility drawings of  diagonal grid graphs and maximal outer 1-plane graphs which are RAC drawable graphs. We also develop algorithms for finding bar 1-visibility drawings of recursive quadrangle 1-planar graphs and pseudo double wheel 1-planar graphs which are not RAC graphs.

The rest of the paper is organized as follows.
Section~\ref{preliminaries} describes some of the definitions that we have used in our paper.
Sections~\ref{RACPB} deals with linear algorithms for finding bar 1-visibility drawings of diagonal grid graphs and maximal outer 1-plane graphs and  Section~\ref{planar and bar} deals with algorithms for finding bar 1-visibility drawings of recursive quadrangle 1-planar graphs and pseudo double wheel 1-planar graphs. 
  Finally, Section~\ref{conclusion}  concludes our paper with a list of open problems. 

\section{Preliminaries}
\label{preliminaries}
 
In this section we introduce some terminologies and  definitions which will be used throughout the paper. 
 For the graph theoretic definitions which have not been described here, see~\cite{DETT99,NR04}.
 
A graph is {\it planar} if it can be embedded in the plane without edge crossing
except at the vertices where the edges are incident. A {\it plane graph} is a
planar graph with a fixed planar embedding. A plane graph divides the plane 
into some connected regions called the {\it faces}. The unbounded region is
called the {\it outer face} and all the other faces are called the {\it inner faces}.
The vertices on the outer face are called the {\it outer vertices} and all
 the other vertices are called {\it inner vertices}.

A \textit{$p \times q$-grid graph} is a graph whose vertices correspond to the grid points of a $p \times q$-grid in the plane and edges correspond to the grid lines between two consecutive grid points. 
\textit{A diagonal grid graph}  $G_{p,q}$ is a $p\times q$-grid graph with diagonal edges are introduced in each cell.
Figure~\ref{figure:Diagonal Grid Graph}(a) shows a $p\times q$-grid graph and Figure~\ref{figure:Diagonal Grid Graph}(b) shows a diagonal grid graph $G_{p,q}$. 
Let $abcd$ be a cell of a diagonal grid graph as illustrated in Figure~\ref{figure:Diagonal Grid Graph}(c), where $a$ is the bottom-left vertex, $b$ is the bottom-right vertex, $c$ is
the up-right vertex and $d$ is the up-left vertex. We call the edge $(a,c)$ the {\it right-diagonal edge} and the
edge $(b,d)$ the {\it left-diagonal edge} of the cell $abcd$.

\begin{figure}[!htbp]
\centering
\includegraphics[width=0.7\textwidth]{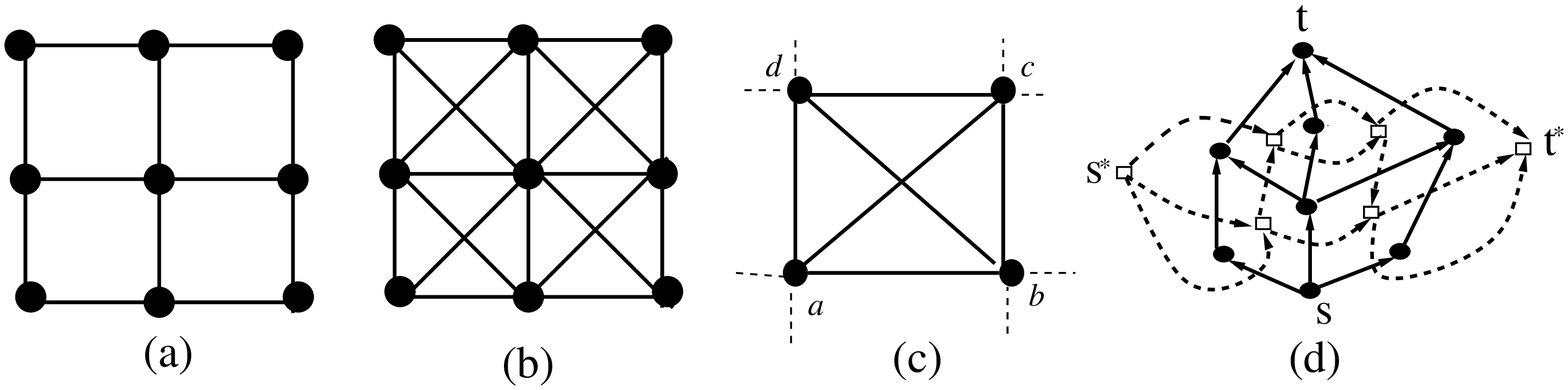}
\caption{(a) Grid graph, (b) diagonal grid graph, (c) one cell of a diagonal grid graph and (d) constructing the dual graph $G^{\ast}$ from planar $st$-graph.}
\label{figure:Diagonal Grid Graph}
\end{figure}

An acyclic digraph with a single source $s$ and a single sink $t$ is called an {\it $st$-graph.} A 
{\it plane $st$-graph} is an $st$-graph that is planar and embedded with vertices $s$ and $t$ on 
the boundary of the outer face. 
Let  $G$ be a plane $st$-graph and $F$ be its set of faces. $F$ contains two representatives
of the outer face: the ”left outer face” $s^{\ast}$, which is incident with the edges on the
left boundary of $G$ and the ”right outer face” $t^{\ast}$, which is incident with the edges on
the right boundary of $G$. Additionally, for each $e =(u, v)$  we define $orig(e) = u$ and
$dest(e) = v$. Also, we define $left(e)$ (respectively $right(e)$) to be the face to the left
(respectively right) of $e$. We now define a dual-like graph $G^*$ of $G$ as follows.  The vertex set of
$G^*$ is the set $F$ of faces of $G$, and $G^*$ has an edge $e^*=(f,g)$ for each edge $e\not= (st)$ of $G$, where 
$f=left(e)$ and $g=right(e)$.
In Figure~\ref{figure:Diagonal Grid Graph}(d) the vertices and edges of $G$ are drawn by black circles and solid lines respectively, and
the vertices and edges of $G^*$ are drawn by white rectangles and dotted lines respectively. 


Let $G$ be a digraph with $n$ vertices and $m$ edges. A {\it topological numbering} of $G$ is an assignment of numbers to 
the vertices of $G$, such that, for every edge $(u,v)$ of $G$, the number assigned to $v$ is greater  than the one assigned 
to $u$. If the edges of digraph $G$ have nonnegative weights associated with them, a {\it weighted topological numbering}
is a topological numbering of $G$, such that, for every edge $(u,v)$ of $G$, the number assigned to $v$ is greater than or
equal to the number assigned to $u$ plus the weight of $(u,v)$. The numbering is {\it optimal} if the range of numbers assigned to the vertices is minimized.

Let $G$ be a planar $st$-graph with $n$ vertices. Two paths $\pi_1$ and $\pi_2$ of $G$ are said to be
non-intersecting if they are edge disjoint and do not cross at common vertices but they can touch at vertices.
Let $\Pi$ be a collection
of non-intersecting paths of $G$. In the visibility drawing of $G$, for every path $\pi$ of $\Pi$, if the edges of $\pi$
are vertically aligned then the drawing is called {\it constrained visibility drawing}.
The following result on constrained visibility drawing is known~\cite{DTT92}.

\begin{lemma}
\label{lemma:eq1}
Let $G$ be a planar $st$-graph with $n$ vertices, and let $\Pi$ be a set of non-intersecting
paths  of $G$. Then one can find a constrained visibility drawing of $G$ in $O(n)$ time with $O(n^2)$ area, where the edges of
every path $\pi$ in $\Pi$ are vertically aligned. 
\end{lemma}
A {\it $1$-planar drawing} of a graph $G$ is a drawing of $G$ on a two dimensional plane where an edge can be crossed by at most another edge. A graph $G$ is {\it $1$-planar} if $G$ has a 1-planar drawing. A 1-planar graph $G$ is {\it optimal} if no edges can be added to $G$ without losing 1-planarity. That is, an optimal 1-planar graph of $n$ vertices has the highest number of edges among all 1-planar graphs of $n$ vertices. An \textit{outer-1-plane graph} is a topological embedding of a graph such that all
vertices lie on the outer face and there is at most one crossing on each edge. An outer-1-plane graph $G = (V, E)$ is a {\it maximal outer-1-plane graph} if for each
pair $u, v$ of vertices where $(u, v)$ is not an edge, adding the edge $(u, v)$ to $G$ makes
it not outer-1-plane; that is, $G^{\prime} = (V, E \cup {(u, v)})$ is not outer-1-plane for every
drawing of the edge $(u, v)$.



\section{Bar 1-Visibility Drawings of 1-Planar RAC Graphs}
\label{RACPB}
Some interesting labeling properties of diagonal grid graphs have been  studied recently by Selvaraju and  Pricilla~\cite{SEL}. Recently 
Dehkordi showed that outer-1-plane graphs are RAC graphs~\cite{HRD}. In this section we give linear-time algorithms for finding bar 1-visibility drawings of diagonal grid graphs and  maximal outer 1-plane graphs which are RAC drawable graphs. In Section~\ref{DGGraphs} we develop an algorithm for finding a bar 1-visibility drawing of a diagonal grid graph and in Section~\ref{MOGraphs} we develop an algorithm for finding a bar 1-visibility drawing of a maximal outer 1-plane graph.

\subsection{Diagonal Grid Graphs}
\label{DGGraphs}
 In this Section we prove the following theorem. 

\begin{theorem}
A bar 1-visibility drawing of a diagonal grid graph can be drawn in linear-time. 
\end{theorem}

\begin{proof}
Let $G_{p,q}$ be a diagonal grid graph as illustrated in Figure~\ref{figure:bbs4}(a). Clearly $G_{p,q}$  is RAC drawable and 1-planar. We will prove that $G_{p,q}$ is bar 1-visible by constructing a bar 1-visibility drawing of $G_{p,q}$. We first obtain a graph $G$ from $G_{p,q}$ by 
deleting  the left-diagonal edge of each cell. Clearly $G$ is a plane graph as illustrated in Figure~\ref{figure:bbs4}(b). Let $v_{ij}$, $1\le i\le p$ and $1\le j\le q$, be the vertex corresponding to the grid point on the $i$th row and $j$th column of the $p\times q$ grid. We now assign a number $Y(v_{ij})$ to each vertex $v_{ij}$ as follows. We set $Y(v_{1,j})=j$ for $1\le j\le q$, and  for $1<i\le p$
we set $Y(v_{ij})= Y(v_{(i-1)j})+2$. We now construct a directed graph from $G$ by assigning direction to each edge from 
lower number to higher number. We add a vertex $s$ below the first row and a vertex $t$ above the last row. We also add directed edges $(s,v_{1,j})$ and $(v_{p,j},t)$ for $1\le j\le q$. Let $G_{st}$ be the resulting digraph as illustrated in Figure~\ref{figure:bbs4}(c). From the construction one can observe that $G_{st}$ is an $st$-graph. We now construct a visibility drawing of $G_{st}$ as follows
\cite{TT86}.

We first construct $G^*_{st}$ of $G_{st}$ and assign unit weights to the edges of $G^*_{st}$ and compute an optimal weighted
topological numbering $X$ of $G^*_{st}$ as illustrated in Figure~\ref{figure:bbs4}(d). We then draw each vertex $v$ as a horizontal line segment $\Gamma(v)$   at $y$-coordinate $Y(v)$ and between $x$-coordinate $X(left(v))$ and $X(right(v)-1)$. We call $X(left(v))$ the start of
$\Gamma(v)$ and $X(right(v)-1)$ the end of $\Gamma(v)$. For each edge $e$, we draw the vertical line segment $\Gamma(e)$ at 
$x$-coordinate $X(left(e))$, between $y$-coordinate $Y(orig(e))$ and $Y(dest(e))$. Let $u$ be the upper-left vertex and $v$ be the bottom-right vertex of a cell in a diagonal grid graph. Then one can observe from the drawing algorithm that $X(right(v)-1) -X(left(u)) = 2$. 

We now obtain a bar 1-visibility drawing of $G$ from the visibility drawing $\Gamma$ of $G_{st}$ as follows.
  We first delete $\Gamma(s)$ and $\Gamma(t)$ from the drawing together with the drawings of the edges incident to
$s$ and $t$. The visibility drawing $\Gamma$ of $G_{st}$ is illustrated in Figure~\ref{figure:bbs4}(e).
We insert one vertical grid line (column) between the two consecutive vertical grid lines $i$ and $j$ if   $i=left(e)$ and $j=right(e)$ for some right diagonal edge $e$  by expanding the drawing towards $+x$ direction. We perform this insertion operation for every $i,j$.  After this insertion operation the difference of $x$-coordinate of end of $\Gamma(v)$ and start
of $\Gamma(u)$ will be three, where $u$ is the upper-left vertex and $v$ is the bottom-right vertex of a cell in a diagonal grid graph. We thus  extend the end of $\Gamma(v)$ by 2 unit in $+x$-direction and the end of  $\Gamma(u)$ by one unit to the $-x$-direction.  We can place the deleted left diagonal edges in the vertical segment which will be placed between starting point of the horizontal bar corresponding to bottom-right vertex and end point of the extended horizontal segment corresponding to up-left vertex in each cell. By extending these bars, the right diagonal edge in each cell crosses horizontal bar corresponding to up-left vertex in the drawing.  Since all edges including left diagonal edges can be placed at end point and start point of the horizontal bars then only right diagonal edges always pass through one horizontal bar corresponding to the vertex, the drawing becomes a bar 1-visibility drawing. The bar 1-visibility drawing of $G_{p,q}$ is illustrated in Figure~\ref{figure:bbs4}(f).
\end{proof}


\begin{figure}[!htbp]
\centering
\includegraphics[width=0.8\textwidth, height=0.7\textheight]{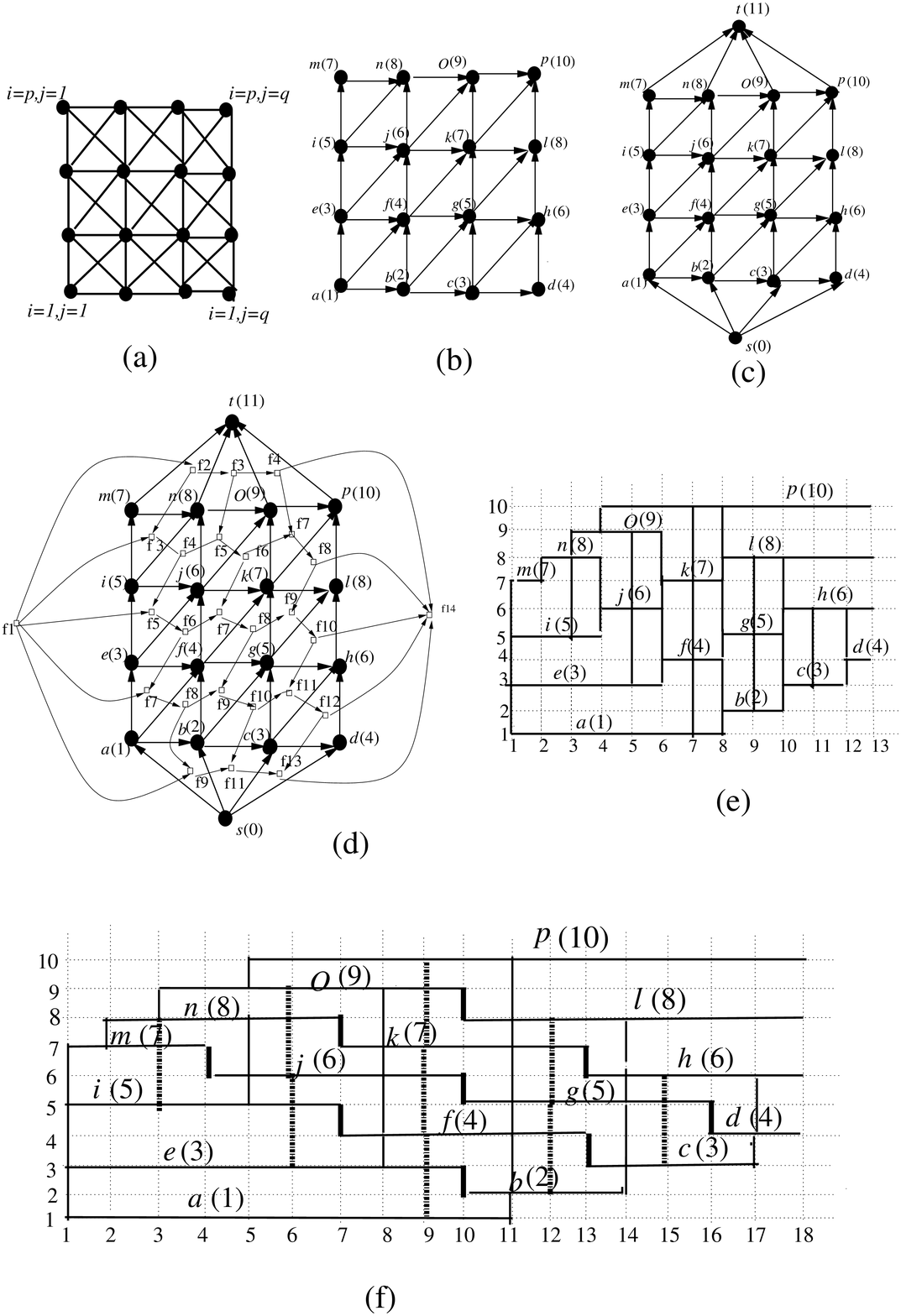}
\caption{(a) Diagonal grid graph $G_{p,q}$, (b) graph $G$ with numbering of vertices and (c) the digraph $G_{st}$, (d) Dual graph $G^*_{st}$, (e) bar visibility drawing and (f) bar 1-visibility drawing of diagonal grid graph $G_{p,q}$.}
\label{figure:bbs4}
\end{figure}

\begin{theorem}
\label{th:dgg} A diagonal grid graph $G_{p,q}$ admits a bar $1$-visibility drawing $\Gamma$  on a grid of size $(q+2p-2)\times(3(p+q)-3)$. Furthermore, $\Gamma$ can be found in linear-time.
\end{theorem}

\begin{proof}

Let $G_{p,q}$ be a diagonal grid graph. We first obtain a graph $G$ from $G_{p,q}$ by 
deleting  the left-diagonal edge of each cell. Then according to proof of theorem 1, we can get visibility drawing of the graph $G$.  The height of the drawing is $Y(v_{pq})=q+2p-2$. The  width of the drawing is the length of
the longest path in $G^*_{st}$. Since each cell of $G$ contains two faces, the longest path of $G^*_{st}$ is at most 
$2(p+q)$. We now compute the size of bar 1-visibility drawing.  One can observe that at most $(p-1)+(q-1)-1= p+q-3$ new
columns are inserted for constructing bar 1-visibility drawing of $G_{p,q}$ from the visibility drawing of $G$. Thus the width of the bar 1-visibility drawing is $(3(p+q)-3)$. Thus we can obtain a  bar $1$-visibility drawing on a grid of size $(q+2p-2)\times(3(p+q)-3)$.  

\end{proof}

A diagonal grid graph $G$ has $n=pq$ vertices. If $ p>q$ then the area of the bar 1-visibility drawing 
 is $O(p^2)$ and if $ p<q$ then the area is $O(q^2)$ . When $ p\cong q$ then the area of the
 bar $1$-visibility drawing  is $O(pq)= O(n)$. The bar 1-visibility drawing obtained by our algorithm is 
``compact" in a sense that there is at least one line segment for every vertical and horizontal grid line except the last
vertical grid line.

\subsection{Maximal Outer 1-Planar Graphs}
\label{MOGraphs}
  In this Section we give an algorithm for obtaining a bar 1-visibility drawing of a
maximal outer 1-plane graph. This problem has an interesting correlation with a constrained visibility drawing of a planar $st$-graph. 
To describe the algorithm, we have need some definitions.

Let $G$ be a maximal outer 1-plane graph. In a maximal outer 1-plane graph, each edge crossing is surrounded by a cycle of length 4 and no edge in this cycle has a crossing~\cite{HRD}. We call such a cycle of length four a {\it quadrangle}. An edge $e$ of an outer-1-plane graph is called a  {\it crossing edge} if $e$ has a crossing; otherwise $e$ is called a {\it non-crossing edge}. In the Figure~\ref{figure:Dlg1}(a), $a,b,c,d$ and $b,c,e,f$ are quadrangles. In  quadrangle $a,b,c,d$, edges $(a,c)$ and $(b,d)$ are  crossing edges and edges $(a,b)$, $(b,c)$, $(c,d)$ , $(d,a)$ are non-crossing edges.
The crossing edges of a quadrangle are called {\it diagonal} of the quadrangle.

We call a labeling of vertices of a maximal outer 1-plane graph by integers 1 to $n$ (where $n$ is the number of vertices in the graph)  a \textit{diagonal labeling} if vertices of each quadrangle got numbers in such a way that two ends of a diagonal of a quadrangle get the lowest and highest numbers among the numbers assigned to the four vertices of the quadrangle. 

For example, Figure~\ref{figure:Dlg1}(a) shows
an input maximal outer 1-plane graph $G$ and Figure~\ref{figure:Dlg1}(b) shows a
diagonal labeling of $G$. 


 We can introduce diagonal labeling on the maximal outer 1-plane graph. We have the following lemma.

\begin{lemma}
\label{lemma:eq}
  Every maximal outer 1-plane graph admits diagonal labeling.
\end{lemma}

\begin{proof}
 Let $G$ be a maximal outer 1-plane graph. We will prove that $G$ has a diagonal labeling. Let $v$ be a vertex in $G$. Then assign 1 to $v$. After that we will give next numbers to those vertices which are incident to non-crossing edges from $v$ in counterclockwise order. Then assign numbers to the vertices which are incident to crossing edges from $v$. We then consider the vertex labeled by 2 and assign next numbers to the vertices in the same way. Since the labeling has done always in increasing order and the vertices incident to crossing edges are labeled later, so the diagonal of a quadrangle got the highest and lowest numbers among the numbers of four vertices of the quadrangle. Figure~\ref{figure:Dlg1}(b) illustrates a
diagonal labeling of the maximal outer 1-plane graph in Figure~\ref{figure:Dlg1}(a). 
\end{proof}

Using this diagonal labeling, we can construct a bar 1-visibility drawing of a maximal outer 1-plane graph as mentioned in the following theorem.

\begin{theorem}
A bar 1-visibility drawing of a maximal outer 1-plane graph can be drawn in linear-time.
\end{theorem}

\begin{proof}
 
Let $G$ be a maximal outer 1-plane graph as illustrated in Figure~\ref{figure:Dlg1}(a). By Lemma~\ref{lemma:eq}, $G$ has a diagonal labeling as illustrated in Figure~\ref{figure:Dlg1}(b). We first give direction to every edge from lower number to higher number as illustrated in Figure~\ref{figure:Dlg1}(b). We next identify the crossing edge containing the highest number and the lowest number in each quadrangle. We next construct a planar graph $G^{\prime}$ from $G$ by passing the crossing edge containing the highest number and the lowest number through the vertex which are right side of the edge in each quadrangle as illustrated in Figure~\ref{figure:Dlg1}(c). The crossing edge which is passed through the vertex is identified as a non-intersecting path. Since we pass all the crossing edges through the vertices which are right side of the edges, these satisfies the conditions of non-intersecting paths stated in the result on constrained visibility drawing~\cite{DTT92}. In the graph illustrated in Figure~\ref{figure:Dlg1}(c), more than one sink vertices are found. We next construct planar $st$-graph by adding dummy edges between the sink vertices to the highest labeled sink vertex as illustrated in Figure~\ref{figure:Dlg1}(d). Since the graph is outer planar, the obtained graph remains planar after adding dummy edges.  
We next construct a constrained visibility drawing of this planar $st$-graph according to Lemma~\ref{lemma:eq1}~\cite{DTT92}. All edges can be placed at end point and start point of the horizontal bars. Since crossing edges containing the highest number and the lowest number for all quadrangles always pass through one horizontal bar corresponding to the vertex, the drawing becomes a bar 1-visibility drawing. Since every step of the algorithm can be done in linear-time, a bar 1-visibility drawing of a maximal outer 1-plane graph can be drawn in linear-time~\cite{DTT92}. The bar 1-visibility drawing of $G$ is illustrated in Figure~\ref{figure:Dlg1}(e).
\end{proof}

\begin{figure}[!htbp]
\centering
\includegraphics[width=0.7\textwidth, height=0.4\textheight]{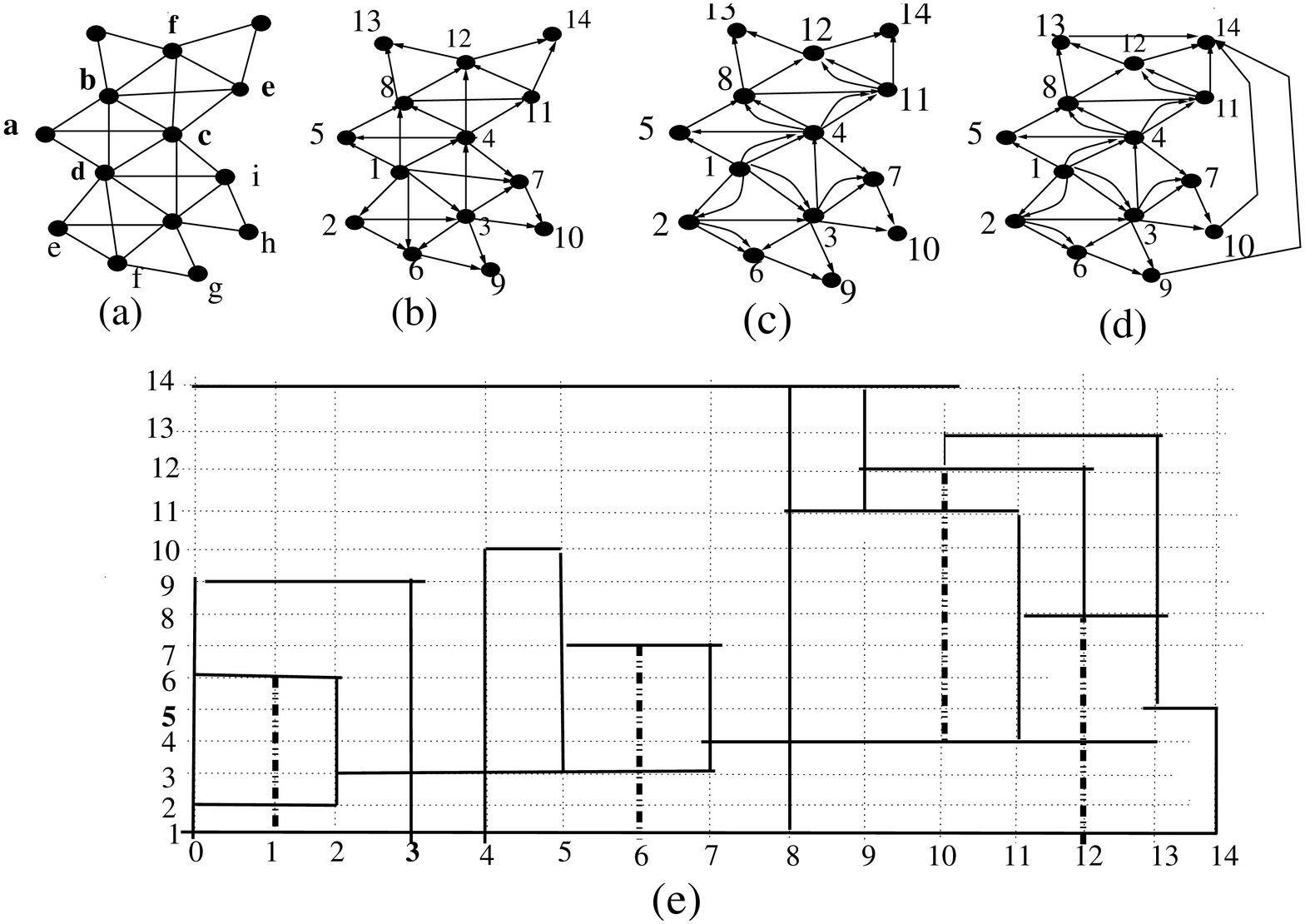}
\caption{(a) A maximal outer 1-plane graph $G$, (b) a diagonal labeling of $G$, (c) planar graph $G^{\prime}$, (d) planar $st$-graph and (e) bar 1-visibility drawing of $G$ }
\label{figure:Dlg1}
\end{figure}

\section{Bar 1-Visibility Drawings of 1-Planar Non-RAC Graphs}
\label{planar and bar}
In the previous section, we showed that diagonal grid graphs and maximal outer 1-planar graphs are two classes of 1-planar graphs.
Recently Eades and Liotta showed that every maximally dense RAC graph is 1-planar \cite{EL12}; on the other hand, they  introduced a class of 1-planar graphs which is not RAC drawable. In this section  we call the class recursive quadrangle 1-planar graphs. Suzuki studied the combinatorial properties of the optimal 1-planar graphs having $4n-8$ edges~\cite{S10} which we will define as ``pseudo double wheel 1-planar graphs". Since a RAC graph with $n>3$ vertices has at most $4n-10$ edges \cite{DEL09}, pseudo double wheel 1-planar graphs are not RAC graphs. In this section we show that recursive quadrangle 1-planar graphs and pseudo double wheel 1-planar graphs are bar 1-visible graphs. 
In Section~\ref{Recursive} we develop an algorithm for finding a bar 1-visibility drawing of a recursive quadrangle 1-planar graph and in Section~\ref{Pseudo} we develop an algorithm for finding a bar 1-visibility drawing of a pseudo double wheel 1-planar graph.

 \subsection{Recursive Quadrangle 1-Planar Graphs }
\label{Recursive}
The class recursive quadrangle 1-planar graph is defined recursively, as follows. $G_0$ is a 1-planar graph of eight vertices . $G_0$ has an 1-planar drawing as illustrated in Figure~\ref{figure:bar1visibility}(a) where the cycle $abcd$ is drawn as the outer rectangle and none of the four edges on the outer rectangle  has a crossing.
Let $abcd$ be the outer rectangle of $G_i$, $i\ge 0$. Graph $G_{i+1}$ is obtained from $G_i$ by adding a new outer rectangle $a^\prime b^\prime c^\prime d^\prime$ and  16 new edges as described 
in Figure~\ref{figure:bar1visibility}(c) and (e), where the four edges on the the outer face do not have any crossing.  

Let $G$ be a 1-planar graph and $x$ be a vertex of $G$ and $(x,y)$ be an edge of $G$. We denote by $\Gamma(G)$, $\Gamma(x), \Gamma(x,y)$ a bar 1-visibility drawing of  $G$, the drawing of a vertex $x$ as a horizontal bar in $\Gamma(G)$, and the drawing of a  an edge $(x,y)$ as a vertical line segment in $\Gamma(G)$, respectively.  For a vertex $x$ in $G$, let $a$ and $b$ be the $x$-coordinates of the two ends of $\Gamma(x)$ such that $a<b$. We call $a$ and $b$ the {\it left end} and the {\it right end} of $\Gamma(x)$, respectively. We now have the following theorem.


\begin{theorem}
\label{th:rq}
Every recursive quadrangle 1-planar graph $G_i,i\ge0$ is a bar 1-visible graph.
\end{theorem}
\begin{proof}
We prove the claim by induction on $i$. Let $abcd$ be the outer rectangle of $G_0$, as illustrated in Figure~\ref{figure:bar1visibility}(a). Then $G_0$ has a bar 1-visibility drawing, as illustrated in Figure~\ref{figure:bar1visibility}(b) where (i) $\Gamma(a)$ is the bottommost bar, $\Gamma(c)$ is the topmost bar, $\Gamma(d)$ is the second bottommost bar and $\Gamma(b)$ is the second topmost bar; (ii) starts of $\Gamma(c)$ and $\Gamma(d)$ have the smallest $x$-coordinate of the drawing and starts of $\Gamma(a)$ and $\Gamma(b)$ have the second smallest $x$-coordinate of the drawing; and (iii) $\Gamma(a,b)$ crosses $\Gamma(d)$.
 We assume that $i>0$ and $G_j$, for  $j<i$, has a bar visibility drawing satisfying (i)-(iii) above. We now 
show that $G_i$ has a bar 1-visibility drawing satisfying (i)-(iii) above. Let $a^\prime b^\prime c^\prime d^\prime$ be the outer rectangle of $G_i$.
We obtain a graph $G_{i-1}$ by deleting the four vertices on 
the outer rectangle  of $G_i$. Let $abcd$ be the outer rectangle of $G_{i-1}$.  
By induction hypothesis, $G_{i-1}$ has a bar 1-visibility drawing satisfying (i)-(iii) 
as illustrated in Figure~\ref{figure:bar1visibility}(d).
We now obtain a bar 1-visibility drawing of $G_i$ from the visibility drawing  of $G_{i-1}$ as follows. 
Extend the left ends of $\Gamma(a)$,  $\Gamma(d)$ and $\Gamma(b)$ by four, four and two units, respectively,  to the $-x$-direction. Extend the right ends of $\Gamma(c)$ and $\Gamma(b)$ by five unit each to the $+x$-direction.
Draw $\Gamma(a^\prime), \Gamma(d^\prime), \Gamma(b^\prime)$ and $\Gamma(c^\prime)$ as the bottommost, 2nd bottommost, 2nd topmost and  topmost bars outside $\Gamma(G_{i-1})$ and draw the new edges as vertical line segments such that (i)-(iii) are satisfied in $\Gamma(G_i)$, as illustrated in Fig.~\ref{figure:bar1visibility}(f).
\end{proof}

\begin{figure}[!htbp]
\centering
\includegraphics[scale=.3]{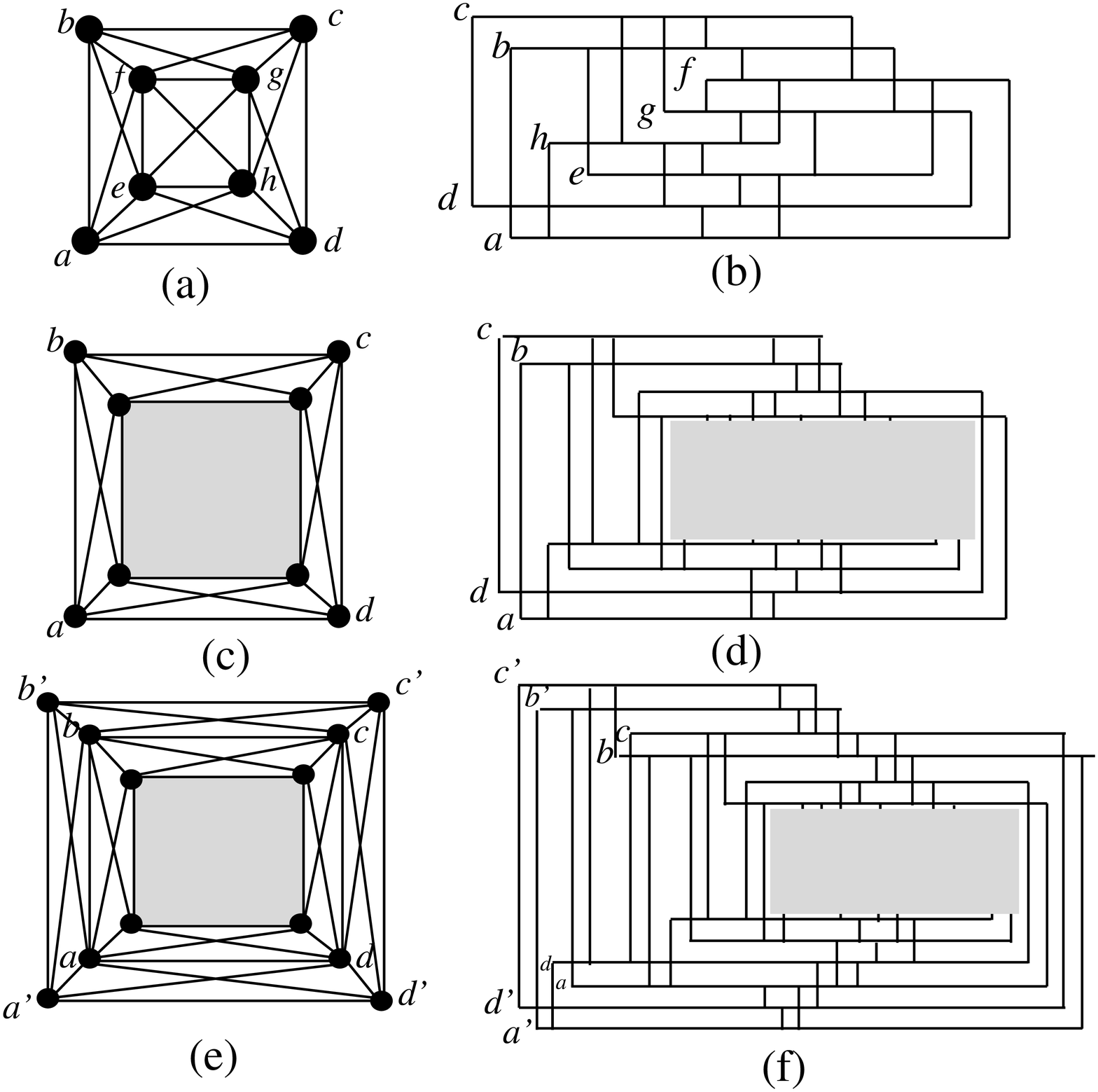}
\caption{Illustration for the proof of Theorem~\ref{th:rq}}
\label{figure:bar1visibility}
\end{figure}

By adding two more edges to recursive quadrangle 1-planar graph $G$ of $4p$, $p\ge 2$ vertices,  we can obtain an optimal 1-planar graph $G^\prime$ of $4p$, $p\ge 2$ vertices. From the bar 1-visibility drawing of $G$ we can obtain a bar 1-visibility 
drawing of $G^\prime$ by adding the drawing of additional two edges as vertical line segment, as illustrated in 
Figure~\ref{figure:bar1optimal}. Thus the following theorem holds.

\begin{figure}[!htbp]
\centering
\includegraphics[scale=.3]{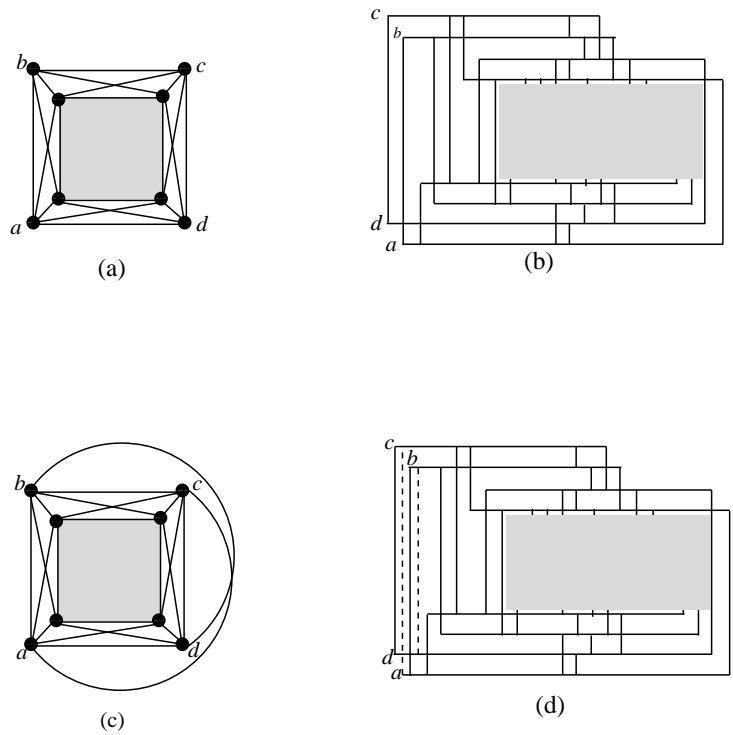}
\caption{(a) A recursive quadrangle 1-planar graph $G$, (b) a bar 1-visibility drawing of $G$, (c) an optimal 1-planar graph $G^\prime$ obtained from $G$ by adding two edges and (d) a  bar 1-visibility drawing of $G^\prime$.}
\label{figure:bar1optimal}
\end{figure}

\begin{theorem}
Every optimal 1-planar graph obtained from recursive quadrangle 1-planar graph by adding two edges is a bar 1-visible graph.
\end{theorem}


\subsection{Pseudo Double Wheel 1-Planar Graphs }
\label{Pseudo}
Let $C$ be a cycle $v_1$$u_1$$v_2$$u_2$...$v_n$$u_n$ of even number of vertices  embedded on a plane. Let $x$ and $y$ be two vertices outside and inside of $C$, respectively. We add $x$ with $u_i$ and $y$ with $v_i$ for $i=1...n$. Let $H$ be the resulting plane graph, as illustrated in Figure~\ref{figure:w2}(a). We add a pair of crossing edges to each face of $H$, as illustrated in Figure~\ref{figure:w2}(b). The resulting graph is an optimal 1-planar graph as introduced by Suzuki~\cite{S10}. We call this optimal 1-planar graph {\it even pseudo double wheel 1-planar graph}.

\begin{figure}[!htbp]
\centering
\includegraphics[scale=.4]{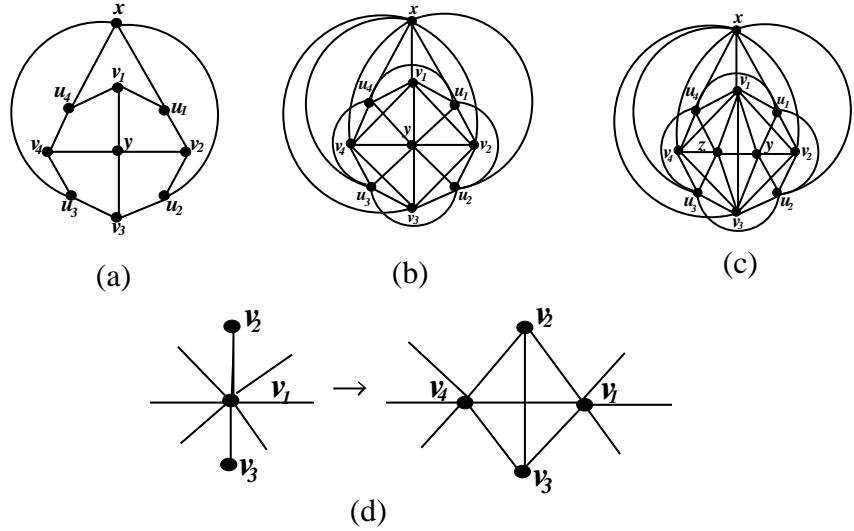}
\caption{(a) Pseudo double wheel, (b) even pseudo double wheel 1-planar graph $G$, (c) odd pseudo double wheel 1-planar graph $G^\prime$ and (d) $Q_v$ splitting.} 
\label{figure:w2}
\end{figure}

{\it $Q_v$ splitting} is an expansion operation at {\it $v_1$} defined as follows : (i) Identify {\it $v_2$} and {\it $v_3$} such that there are {\it $v_1$$v_2$} and {\it $v_1$$v_3$} edges but no {\it $v_2$$v_3$} edge. (ii) Split $v_1$$v_2$$v_3$ path (iii) Rename one copy of $v_1$ as $v_4$. (iv) Join $v_1$ , $v_4$ and $v_2$ , $v_3$. Splitting is illustrated in Figure~\ref{figure:w2}(d). An optimal 1-planar graph obtained from even pseudo double wheel 1-planar graph by one splitting operation is called {\it odd pseudo double wheel 1-planar graph} as illustrated in Figure~\ref{figure:w2}(c).
 We now have the following theorem. 

\begin{theorem}
\label{th:opt}
Every pseudo double wheel 1-planar graph is a bar 1-visible graph.
\end{theorem}

\begin{proof}
We will prove every even pseudo double wheel 1-planar graph and odd pseudo double wheel 1-planar graph  is a bar 1-visible graph. 

Let $G$ be an even pseudo double wheel 1-planar graph. We will prove that $G$ is a bar 1-visible graph by constructing a bar 1-visibility drawing of $G$. We first  draw $n$ bars where $1st$ topmost bar is $y$, $2nd$ topmost bar is $x$ and other n-2 bars starting from $3rd$ topmost bar are $v_1$, $u_1$, $v_2$, $u_2$...$v_n$, $u_n$ respectively. We next draw $n-2$ vertical lines from  $v_1$, $u_1$, $v_2$, $u_2$...$v_n$, $u_n$ to $y$ each crossing $x$ and another $n-2$ vertical lines from $v_1$, $u_1$, $v_2$, $u_2$...$v_n$, $u_n$ to $x$ without crossing any bar. We next join each of the $v_2$, $u_2$, $v_3$, $u_3$...$v_{n-1}$, $u_{n-1}$ bars to 4 bars 1 and 2 unit up and below itself. We next consider $v_1$, $u_1$, $v_n$ and $u_n$. We next join $v_1$ to $u_1$, $v_1$ to $v_n$, $v_1$ to $u_n$ and $v_n$ to $u_n$ by vertical lines without crossing any bar. At last we draw a vertical line from $u_1$ to  $u_n$ crossing $v_n$. Since each vertical line crosses at most one bar, the drawing becomes bar 1-visibility drawing. The bar 1-visibility drawing of $G$ is illustrated in Figure~\ref{figure:w1}(a).

Let $G^{\prime}$ be an odd pseudo double wheel 1-planar graph of $n$ vertices. We will prove that $G^{\prime}$ is a bar 1-visible graph by constructing a bar 1-visibility drawing of $G^{\prime}$. We first find $v_i$ and $v_j$ on C such that $v_i$ and $v_j$ have degree 8. We next draw $n$ bars where 1st topmost bar is $z$, 2nd topmost bar is $y$, 3rd topmost bar is $x$ and other $n-3$ bars starting from 4th topmost bar are $v_1$, $u_1$, $v_2$, $u_2$...$v_n$, $u_n$ respectively. We next draw $n-3$ vertical lines from $v_1$, $u_1$, $v_2$, $u_2$...$v_n$, $u_n$ to $x$ without crossing any bar. Then join $y$ to $v_i$,$u_i$, $v_{i+1}$, $u_{i+1}$...$v_j$ by vertical lines. These lines cross bar $x$. We next draw vertical lines from $z$ to $v_j$,$u_j$, $v_{j+1}$, $u_{j+1}$...$v_n$, $u_n$, $v_1$, $u_1$...$v_j$ crossing bar $x$. We next join each of the $v_2$, $u_2$, $v_3$, $u_3$...$v_{n-1}$, $u_{n-1}$ bars to 4 bars 1 and 2 unit up and below itself. Now we consider $v_1$, $u_1$, $v_n$ and $u_n$. We next join $v_1$ to $u_1$, $v_1$ to $v_n$, $v_1$ to $u_n$ and $v_n$ to $u_n$, $y$ to $z$ by vertical lines without crossing any bar. We next draw a vertical line from $u_1$ to  $u_n$ crossing $v_n$. At last we join $v_i$ to $v_j$. Since each vertical line crosses at most one bar, the drawing becomes bar 1-visibility drawing.  The bar 1-visibility drawing of $G^{\prime}$ is illustrated in Figure~\ref{figure:w1}(b). 

\end{proof}

\begin{figure}[!htbp]
\centering
\includegraphics[scale=.3]{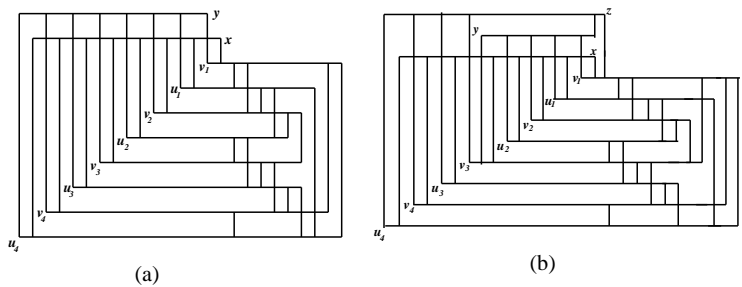}
\caption {(a) Bar 1-visibility drawing of even pseudo double wheel 1-planar graph and (b) bar 1-visibility drawing of odd pseudo double wheel 1-planar graph.} 
\label{figure:w1}
\end{figure}

\section{Conclusion}
\label{conclusion}

In this paper we give linear-time algorithms to find bar 1-visibility drawings of diagonal grid graphs and maximal outer 1-planar graphs which are RAC drawable graphs.  We also developed algorithms for finding bar 1-visibility drawings of recursive quadrangle 1-planar graphs and pseudo double wheel 1-planar graphs which are not RAC drawable graphs.

Pach and Toth \cite{PT97} proved that a 1-planar graph can have at most $4n-8$ edges whereas Dean {\it et al.}~\cite{DEGLST07} showed that a bar 1-visible graph can have at most 
$6n-20$ edges. The bound for bar 1-visible graph is tight since for each $n\ge 8$ there exists a bar 1-visible graph with exactly $6n-20$ edges ~\cite{DEGLST07}. Thus not all bar 1-visible graphs are 1-planar graphs. This can be well illustrated by the following example. A bar 1-visibility drawing of an  optimal 1-planar graph $G$ of eight vertices in 
Figure~\ref{figure:8vertex}(a), is shown in Figure~\ref{figure:8vertex}(b), but Figure~\ref{figure:8vertex}(c) shows a bar 1-visibility drawing of a graph of eight vertices which has  more edges. Suzuki~\cite{S10} proved that every optimal 1-planar graph can be obtained from a even pseudo double wheel 1-planar graph by a sequence of $Q_v$ splittings and ``$Q_4$ additions".
We were able to construct bar 1-visibility drawing of every  1-planar graph  that we considered as an example, but yet to find a formal proof. We thus conjecture as follows.
\begin{figure}[!htbp]
\centering
\includegraphics[scale=.2]{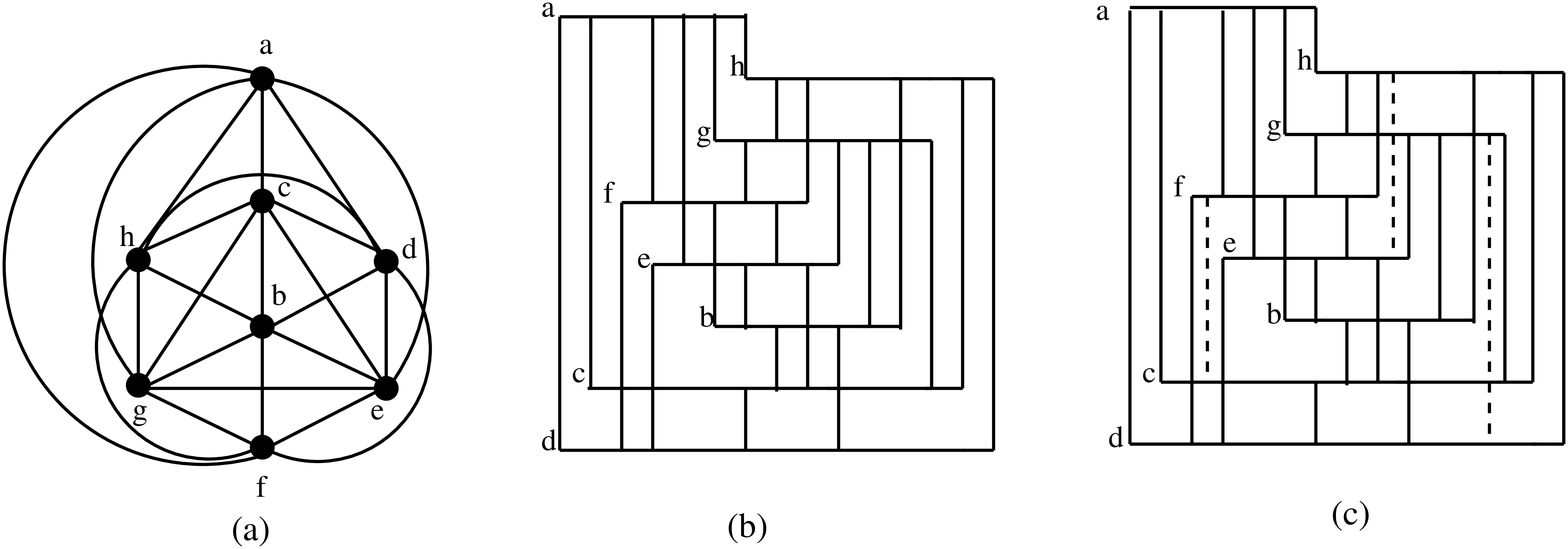}
\caption{(a) An optimal 1-planar graph $G$, (b) a bar 1-visibility drawing of $G$ and (c) a bar 1-visibility drawing of $G^\prime$ obtained by adding some edges to $G$.}
\label{figure:8vertex}
\end{figure}

\begin{conjecture}
Every 1-planar graph is a bar 1-visible graph.
\end{conjecture}
Several interesting open problems have come out from this works.

\begin{enumerate}
\item Recognition of both RAC graphs and 1-planar graphs are NP-complete problems. It is interesting to know the complexity of recognizing a bar 1-visible graph. Finding a complete characterization of bar 1-visible graph is also an interesting 
open problem.
\item How to find a 1-planar embedding of a 1-planar graph?
\item Can we find a complete characterization of bar $k$-visibility drawing?
\end{enumerate}

\bibliographystyle{splncs03}
 \bibliography{barvisibility}



\end{document}